\documentclass[superscriptaddress,showpacs,prl,twocolumn]{revtex4}
\usepackage{amssymb,amsbsy,amsmath,latexsym,dsfont,braket,bm,slashed,graphicx,mathrsfs,color}
\usepackage{amsthm}

\theoremstyle{plain}
\newtheorem{theorem}{Proposition}
\theoremstyle{definition}
\newtheorem{definition}{Definition}

\begin{document}
	
	\title{Cooperative light scattering in any dimension}
	
	\author{Tyler Hill}
	\affiliation{%
		Physics Department, University of Michigan, Ann Arbor}
	\author{Barry C. Sanders}
	\affiliation{%
		Institute for Quantum Science and Technology, University of Calgary, Alberta, Canada T2N 1N4
	}
	\affiliation{%
		Program in Quantum Information Science,
		Canadian Institute for Advanced Research,
		Toronto, Ontario M5G 1Z8, Canada
	}
	\affiliation{Hefei National Laboratory for Physical Sciences at the Microscale,
		University of Science and Technology of China, Hefei, Anhui, China}
	\affiliation{Shanghai Branch, CAS Center for Excellence and Synergetic Innovation Center in Quantum Information and Quantum Physics, University of Science and Technology of China, Shanghai, China}
	\author{Hui Deng}
	\affiliation{%
		Physics Department, University of Michigan, Ann Arbor}
	
%	\date{\today}
	
	\begin{abstract}
	We present a theory of cooperative light scattering valid in any dimension: connecting theories for an open line, open plane, and open space in the non-relativistic regime. This theory includes near-field and dipole-orientation effects, highlighting how field mode confinement controls the phenomena. We present a novel experimental implementation for planar collective effects.
	\end{abstract}
	
	
	\maketitle
	
	Interatomic dipole-dipole coupling yields remarkable collective effects such as super- and sub-radiant emission~\cite{Dic54,FS90,DB96,Bel2014}, Anderson localization~\cite{Ski2014, Maximo2015}, and collective Lamb shifts~\cite{Lehmberg1970}, 
	which test fundamentals of quantum electrodynamics (QED)
	and have applications to superradiant lasers~\cite{BCW+12}, quantum simulation~\cite{GCCK15}, and protecting quantum information~\cite{Lidar2003}. Waveguide quantum electrodynamics enables improved spatial mode matching compared to three-dimensional ($3$D) systems~\cite{Meir2014}, thereby increasing photon-mediated coupling between distant atoms in one-dimensional ($1$D)~\cite{Kuz1983,Gonzalez-Tudela2011,ZB13,Lalumiere2013,VanLoo2013,Sipahigil2016} and two-dimensional ($2$D) systems~\cite{Maximo2015,GCCK15}. We present an elegant unified model for cooperative light scattering by $N$ two-level atoms in an open spatial region of arbitrary dimension $d$, providing a single expression for the collective effects in terms of ``cardinal'' Bessel functions. We propose a scheme to observe the phenomena in 2D using vacancy centers in diamond.
	
	We develop a theory of multi-atom superradiance for electromagnetic fields confined to $d$D ($d\in[1,2,3]$). We solve the collective Lamb shifts and spontaneous emission rates as a function of dimension $d\in[1,2,3]$, dipole orientation, and dipole-dipole separation. We find that orientation effects are especially prominent at small atom-atom separations as dimension increases. Our theory provides intuition into how superradiance can be controlled via field confinement, orientation, and placement of dipoles in realistic structures such as our proposed diamond vacancy center scheme.
	
	In our theory we find that 2D has the most complex orientation dependence between dipoles with subwavelength separations. This complex dependence is due to the lack of cylindrical symmetry with respect to the separation between dipoles, different from both 3D and 1D. Vacancy centers in diamond allow for subwavelength positioning of centers~\cite{Toyli2010,Schukraft2016,Sipahigil2016,Rogers2014} where the orientation-effects are especially prominent.
	
	Our physical system comprises identical two-level systems (here called ``atoms'') coupled to electromagnetic fields propagating in vacuum.
	\begin{figure}
		\includegraphics[width=0.94\columnwidth]{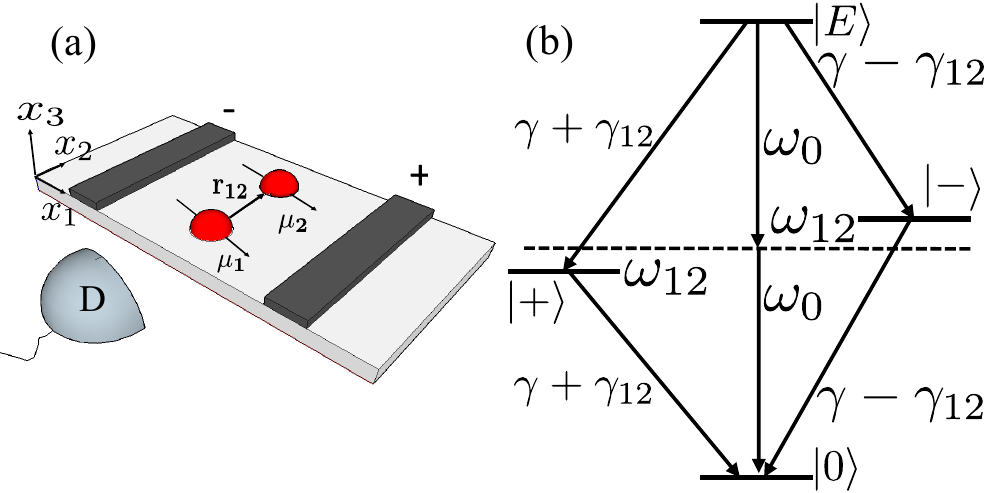}
		\caption{%
			(a) Schematic showing a pair of emitters embedded in a 2D slab
			extending in the $x_1 x_2$ plane.
			The emitters are separated a distance $\bm{r}_{\imath\jmath}$ apart in the $\hat{\bm{x}}_2$ direction.
			Emission is detected by a detector~D. % parallel to the dipoles displaced in the~$\hat{\bm{x}}_2$ direction.
%			A pair of electrodes are attached on top of the slab, which generate the field in the $\hat{\bm x}_1$ direction.
			(b) Energy diagram for $2$-atom superradiance, with $|0\rangle=|g\rangle_1|g\rangle_2$, $|E\rangle=|e\rangle_1|e\rangle_2$, and the superradiant and subradiant states $|\pm\rangle=\frac{1}{2}(|e\rangle_1|g\rangle_2\pm|g\rangle_1|e\rangle_2)$. $|\pm\rangle$ have transition energies $\omega_0\mp\omega_{12}$ and rates $\gamma\pm\gamma_{12}$, as labeled in diagram.}
		\label{fig:schematic}
	\end{figure}
	For a $d$D system, the fields are described by a plane-wave decomposition with wavevector $\bm{k}\in\mathbb{R}^d$ and dispersion $\omega_{\bm{k}} =c|\bm {k}|$. In this work a vector $\bm{a}=\sum_{l=1}^{3}x_l\hat{\bm{x}}_l\in\mathbb{R}^d$ if $\bm{a}\cdot\mathds{1}_d=\bm{a}=\sum_{l=1}^{d}x_l\hat{\bm{x}}_l$, where $\mathds{1}_d$ is the $d$D unit dyad
	$\sum_{l=1}^d\hat{\bm{x}}_l\hat{\bm{x}}_l$,
	which projects vectors into $d$D
	for~$\{\hat{\bm{x}}_l\}$ the orthogonal Cartesian unit vectors.
	
	We solve a master equation describing the evolution of atom states in our system, so following Lehmberg~\cite{Lehmberg1970} we quantize the electromagnetic field. We consider the field quantized in a volume $V$, with photon creation operator~$\hat{a}^{\dagger}_{\bm {k} l}$ producing a photon with wavevector $\bm {k}$, frequency $\omega_{\bm{k}}$, and polarization $\hat{\bm{e}}_l$, $\hat{\bm {k}}\cdot\hat{\bm{e}}_l=0$. We can write the fields as
	\begin{equation}
	\begin{Bmatrix}
	\hat{\bm E}(\bm{r})\\ \hat{\bm B}(\bm{r)}
	\end{Bmatrix}
	=\sum_{\bm k}\sum^{2}_{l=1}
	\sqrt{\frac{2\pi\omega_{\bm {k}}}{V}}
	\begin{Bmatrix}
	\hat{\bm e}_l\\ \hat{\bm k}\times\hat{\bm e}_l
	\end{Bmatrix}
	\left(\text{e}^{\text{i}\bm{k}\cdot\bm{r}}\hat{a}_{\bm {k}l}+\operatorname{hc}\right)\label{eq:quantize}
	\end{equation}
	at point~$\bm r$
	with~hc denoting the hermitian conjugate and~$\hat{}$ denoting operator or unit vector
	(which case pertains is discernible from the context).
	
	Identical atoms are placed at positions $\bm{r}\in\mathbb{R}^d$. We label atoms with indices $\imath$ and $\jmath$
	so that for atom $\imath$ energy~$\hbar\omega_0$ separates its
	excited state~$\ket{\text e}_\imath$ from ground state~$\ket{\text g}_\imath$,
	and the atomic dipole moment~$\bm\mu_\imath$ can be oriented in any direction in~$\mathbb{R}^3$.
	Henceforth $\hbar\equiv1$.
	De-exciting and exciting the atom is achieved by operators
	$\hat{\sigma}_\imath=\ket{\text g}_\imath\!\bra{\text e}$ and $\hat{\sigma}^\dagger_\imath$, respectively.
	\begin{theorem}
		The vacuum expectation of any self-adjoint $N$-atom operator~$\hat Q$ for times $\omega_0 t\gg1$ is
		\begin{align}
		\dot{\hat{Q}}
		=&\sum_{\imath\jmath}^N\text{i}\omega_{\imath\jmath}
		\left[\hat{\sigma}_\imath^{\dagger}\hat{\sigma}_{\jmath},\hat{Q}\right]
		\nonumber\\
		&+\frac{\gamma_{\imath\jmath}}{2}\left(2\hat{\sigma}_\imath^{\dagger}
		\hat{Q}\hat{\sigma}_{\jmath}-\hat{\sigma}_\imath^{\dagger}\hat{\sigma}_{\jmath}\hat{Q}
		-\hat{Q}\hat{\sigma}_\imath^{\dagger}\hat{\sigma}_{\jmath}\right)\label{eq:Q2}
		\end{align}
		for $\omega_{\imath\imath}:=\omega_0$, and
		\begin{align}
		\omega_{\imath\jmath}
		=&-\frac{2\pi }{\text{c}^d}\int \slashed{\text{d}}^{d-1}\Omega_{\hat{\bm{k}}}\bm{\mu}_\imath
		\cdot\left[\mathds{1}_{3}-\hat{\bm k}\hat{\bm k}\right]\cdot\bm{\mu}_{\jmath}
		\nonumber\\&\times
		\sum_{\pm}\mathcal{P} \int_0^{\infty} \slashed{\text{d}}{\omega}
		\frac{\omega^d}{\omega\pm\omega_0} \text{e}^{\text{i} \omega \hat{\bm k}\cdot \bm{r}_{\imath\jmath}/\text{c}},
		\label{eq:omegaA}\\
		%\end{align}
		%\begin{equation}
		\gamma_{\imath\jmath}
		=&\frac{2\pi  \omega_0^d}{c^d} \int\slashed{\text{d}}^{d-1} \Omega_{\hat{\bm{k}}}\bm{\mu}_\imath
		\cdot \left[\mathds{1}_{3}-\hat{\bm k}\hat{\bm k}\right]
		\cdot \bm{\mu}_{\jmath} \text{e}^{\text{i}\omega_0
			\hat{\bm k}\cdot\bm{r}_{\imath\jmath}/\text{c}},
		\label{eq:gammaij}
		\end{align}
		with~$\mathcal{P}$ denoting principle value,
		$\bm{r}_{\imath\jmath}:=\bm{r}_\imath-\bm{r}_{\jmath}$,
		$\slashed{\text{d}}^d:=\text{d}^d/(2\pi)^d$,
		$\text{d}^{d-1}\Omega_{\hat{\bm{k}}}$ the $d$D solid angle integrating over directions $\hat{\bm{k}}$.
	\end{theorem}
	\begin{proof}
		The Hamiltonian for~$N$ identical atoms (with individual frequency~$\omega_0$) coupled to the field is
		\begin{align}
		\hat{H}
		=&\sum_{\imath=1}^N\omega_0\hat{\sigma}_\imath^\dagger\hat{\sigma}_\imath
		+\sum_{\bm {k} l}\omega_{\bm {k}}\hat{a}_{\bm {k} l}^\dagger\hat{a}_{\bm {k} l}					
		-\sum_{\imath=1}^N\sum_{\bm {k} l}
		\left(\frac{2\pi\omega_{\bm {k}}}{V}\right)^{1/2}
		\nonumber\\&
		\times\hat{\bm{e}}_{l}\cdot\bm{\mu}_\imath
		\left(\text{e}^{\text{i} \bm{k}\cdot\bm{r}_\imath}\hat{a}_{\bm {k} l}+\text{hc}\right)
		\left(\hat{\sigma}_\imath+\hat{\sigma}_\imath^\dagger\right).
		\label{eq_H}
		\end{align}
		The quantum master equation for~$\hat Q$ any $N$-atom operator was originally solved for 3D fields by treating atoms as point dipoles and neglecting strong fields and non-local effects~\cite{Lehmberg1970},
		and recently the master equation was solved for 1D fields~\cite{Lalumiere2013}.
		Here we employ the Markovian approximation and solve for it in $d$D
		with $d\in[1,2,3]$ when the time of flight across the sample is faster than any spontaneous emission rate so that non-local effects may be neglected.
		
		We first eliminate the photon operators $\hat{a}_{\bm{k} l}(0)$ which represents the field amplitude of the excitation source.
		We rewrite it in terms of atomic operators using
		\begin{align}
		\hat{a}_{\bm{k} l}(t)=&\hat{a}_{\bm{k} l}(0)\text{e}^{-\text{i}\omega_{\bm{k} }t}+i\sum_\imath\left(\frac{2\pi\omega_{\bm{k} }}{V}\right)^{1/2}\hat{e}_{l}\cdot\bm{\mu}_\imath\text{e}^{-\text{i}\bm{k}\cdot\bm{r}_\imath}
		\nonumber\\
		&\times\int_{0}^{t}\text{d}t'\left[\hat{\sigma}_\imath(t')+\hat{\sigma}_\imath^{\dagger}(t')\right]\text{e}^{-\text{i}\omega_{\bm{k} l}(t-t')}.\label{eq:photon}
		\end{align}
		We then take vacuum expectation values of the master-equation solution to obtain
		\begin{align}
		\label{eq:Q1}
		\dot{\hat{Q}}
		=&\text{i}\omega_0\sum_\imath
		\left[\hat{\sigma}_\imath^\dagger\hat{\sigma}_\imath,\hat{Q}\right]
		+\frac{1}{V}\sum_{\imath\jmath}
		\left[\sigma_\imath+\hat{\sigma}_\imath^\dagger, \hat{Q}\right]
		\nonumber\\&\times
		\Bigg\{\sum_{\bm {k} l} 2\pi\omega_{\bm{k}}(\hat{\bm e}_{l}\cdot\bm{\mu}_\imath)
		\left(\hat{\bm e}_{l}\cdot\bm{\mu}_\jmath\right)
		\text{e}^{\text{i}\bm {k}\cdot\bm{r}_{\imath\jmath}}
		\nonumber\\&\times
		\left[f_-\hat{\sigma}_\jmath+f_+\hat{\sigma}^{\dagger}_\jmath\right]+\text{hc}\Bigg\}
		\end{align}
		with $f_{\pm}= -i\mathcal{P}(\omega\pm\omega_0)^{-1}+\pi\delta(\omega\pm\omega_0)$.
		
		We then express the master equation in terms of collective frequency shifts and corresponding linewidths,
		which involves converting the sum over ${\bm k}$ into integration over $\omega({\bm k})$
		using the dispersion relation $\omega =\text{c}\left|{\bm k}\right|$ and obtain
		\begin{align}
		\frac{1}{V}\sum_{\bm{k}}
		\to& \int\slashed{\text{d}}^d\bm{k}
		\to\frac{1}{\text{c}^d}\int\slashed{\text{d}}\omega\omega^{d-1}\int \slashed{\text{d}}^{d-1}\Omega_{\hat{\bm{k}}},
		\label{eq:modes}
		\\
		\text{d}^{d-1}\Omega_{\hat{\bm{k}}}=&\prod_{l=1}^{d-1}\sin^{d-l-1}\theta_{l}\text{d}\theta_{l}.
		\end{align}
		Here $\text{d}^{d-1}\Omega_{\hat{\bm{k}}}$ is the $d$D solid angle over directions $\hat{\bm{k}}$
		with azimuthal angles $\theta_1,\dots,\theta_{d-2}\in[0,\pi]$ and polar angle $\theta_{d-1}\in[0,2\pi)$.
		Substituting
		\begin{equation}
		\sum_{l=1}^2(\hat{\bm e}_{l}\cdot\bm{\mu}_\imath)
		\left(\hat{\bm e}_{l}\cdot\bm{\mu}_{\jmath}\right)
		=\bm{\mu}_\imath\cdot(\mathds{1}_{3}-\hat{\bm k}\hat{\bm k})\cdot\bm{\mu}_{\jmath},
		\label{eq:kdotmu}
		\end{equation}
		and Eq.~(\ref{eq:modes})
		into Eq.~(\ref{eq:Q1}) completes the proof.
	\end{proof}
	For $N=1$ atom and a $d$D field,
	with~$k_0:=\omega_0/\text{c}=2\pi/\lambda_0$
	and
	$\bm\mu_\imath:=\mu_\imath\hat{\bm{\mu}}_\imath$,
	Eq.~(\ref{eq:gammaij})
	yields spontaneous emission rate
	\begin{align}
	\gamma_{\imath\imath}
	=&\frac{2^{3-d}\pi^{2-d/2}\mu_\imath^2k_0^d}{\Gamma(d/2)}
	\left(1-\frac{\hat{\bm{\mu}}_\imath
		\cdot\mathds1_d\cdot\hat{\bm{\mu}}_\imath}{d}\right)
	\end{align}
	for $\Gamma$ the Gamma function.
	In 3D,
	$\gamma_{\imath\imath}
	=4\mu_\imath^2 k_0^3/3$ is independent of dipole orientation.
	In~1D and 2D,
	$\gamma_{\imath\imath}$ is maximized for the dipole perpendicular to the $\mathbb{R}^d$ subspace
	($\hat{\bm{\mu}}_\imath\cdot\mathds1_d\cdot\hat{\bm{\mu}}_\imath=0$)
	and thus falls by half for in-plane dipoles in 2D
	($\hat{\bm{\mu}}_\imath\cdot\mathds1_2\cdot\hat{\bm{\mu}}_\imath=1$)
	compared to out-of-plane dipoles~\cite{Maximo2015}
	and is zero for in-line dipoles in 1D.
	
	For $r_{\imath\jmath}\ll\lambda$, Eq.~(\ref{eq:omegaA}) is divergent and cannot be used to calculate the single-atom Lamb shift.
	The breakdown of this theory to describe the single-atom Lamb shift is a consequence of approximating a physical dipole with a point dipole.
	We thus treat the single-atom Lamb shift as being incorporated into a renormalized frequency~$\omega_0$.
	
	For $N \ge 2$ atoms,
	signatures of collective-effects, such as enhanced spontaneous decay and Lamb shifts,
	are quantified by~$\gamma_{\imath\jmath}$ and~$\omega_{\imath\jmath}$
	($\imath\neq\jmath$), respectively,
	as illustrated in Fig.~\ref{fig:schematic}(b)
	for $N=2$ atoms.
	We now express~$\gamma_{\imath\jmath}$ and~$\omega_{\imath\jmath}$
	in terms of the~$d$D dyadic Green's function.
	\begin{definition}
		The dyadic Green's function in $d$D is
		$\overleftrightarrow{G}_d:=\mathcal{D}G_d$
		for $\mathcal{D}:=\mathds{1}_{3}+\frac{1}{k_0^2}\nabla_d\nabla_d$
		a dyadic operator,
		$G_d$ the solution of the~$d$D Helmholtz equation $\left[\nabla_d^2+k_0^2\right]G_d\left(\bm{r}_{\imath\jmath},\omega_0\right)=-\delta\left(\bm{r}_{\imath\jmath}\right)$.
	\end{definition}
	\begin{definition}
		Analogous to the relation between $\sin x$ and $\operatorname{sinc}x$ (``cardinal sine''),
		we introduce ``cardinal'' versions of the Bessel functions (first and second kind)
		and Hankel function of the first kind as, respectively,
		\begin{equation*}
		\check{J}_\alpha(x):=\frac{J_\alpha(x)}{x^\alpha},\;
		\check{Y}_\alpha(x):=\frac{Y_\alpha(x)}{x^\alpha},\;
		\check{H}^{(1)}_\alpha(x):=\frac{H^{(1)}_\alpha(x)}{x^\alpha}.
		\end{equation*}
	\end{definition}
	\begin{theorem}
		The complex collective frequency shift is
		\begin{equation}
		\label{eq:Gamma}
		\Gamma_{\imath\jmath}
		:=-\omega_{\imath\jmath}+\text{i}\gamma_{\imath\jmath}/2
		=4\pi  k_0^2\bm{\mu}_\imath
		\cdot\overleftrightarrow{G}_d(\bm{r}_{\imath\jmath},\omega_0)
		\cdot \bm{\mu}_{\jmath}.
		\end{equation}
	\end{theorem}
	\begin{proof}
		Solutions of the $d$D Helmholtz equation are~\cite{Sti77}
		$A\check{J}_{d/2-1}\left(\tilde{r}_{\imath\jmath}\right)
		+B\check{Y}_{d/2-1}\left(\tilde{r}_{\imath\jmath}\right)$
		for~$\tilde{\bm r}_{\imath\jmath}
		:=k_0\bm{r}_{\imath\jmath}
		=\tilde{r}_{\imath\jmath}\hat{\tilde{\bm r}}_{\imath\jmath}$
		and~$A$ and~$B$ arbitrary complex constants.
		Imposing the Sommerfeld radiation condition
		\begin{align}
		\underset{\tilde{r}_{\imath\jmath}\to\infty}{\text{lim}}{\left|\bm{r}_{\imath\jmath}\right|^{(d-1)/2}
			\left(\frac{\partial}{\partial{\tilde r}_{\imath\jmath}}-\text{i}\right)
			G_d\left(\bm{r}_{\imath\jmath},\omega_0\right)=0}.
		\end{align}
		on an outgoing spherical wave satisfying energy conservation yields the purely radial expression
		\begin{equation}
		G_d\left(\bm{r}_{\imath\jmath},\omega_0\right)
		=\text{\ensuremath{\frac{\rm i}{4}}}\left[\frac{k_0^2}{2\pi}\right]^{d/2-1}
		\check{H}_{d/2-1}^{(1)}\left(\tilde{r}_{\imath\jmath}\right).
		\label{eq:DGF2d}
		\end{equation}
		For
		$G'_d:=\frac{\text{d}G_d}{\text{d}\tilde{r}_{\imath\jmath}}$ and $G''_d:=\frac{\text{d}^2G_d}{\text{d}\tilde{r}_{\imath\jmath}^{2}}$,
		applying~$\mathcal D$ to~$G_d$~(\ref{eq:DGF2d}) yields
		\begin{equation}
		\frac{1}{k_0^2}\nabla_d\nabla_dG_d
		=\hat{\bm{r}}_{\imath\jmath}\hat{\bm{r}}_{\imath\jmath}G''_d
		+\frac{\nabla_d\hat{\bm{r}}_{\imath\jmath}}{k_0}G'_d.
		\end{equation}	
		We apply the identity
		\begin{equation}
		\frac{\nabla_d\hat{\bm{r}}_{\imath\jmath}}{k_0}=\frac{1}{\tilde{r}_{\imath\jmath}}(\mathds1_d-\hat{\bm{r}}_{\imath\jmath}\hat{\bm{r}}_{\imath\jmath})
		\end{equation}
		to obtain
		\begin{equation}
		\frac{1}{k_0^2}\nabla_d\nabla_d G_d
		=\frac{1}{\tilde{r}_{\imath\jmath}}G'_d\mathds{1}_d+\left(G''_d
		-\frac{1}{\tilde{r}_{\imath\jmath}}G'_d\right)\hat{\bm{r}}_{\imath\jmath}\hat{\bm{r}}_{\imath\jmath}.\label{eq:dyadop}
		\end{equation}
		Hankel function recurrence relations then yield
		\begin{align}
		\label{eq:GdH}
		\overleftrightarrow{G}_d\left(\tilde{\bm{r}}_{\imath\jmath},\omega_0\right)
		=&\text{\ensuremath{\frac{\rm i}{4}}}\left[\frac{k_0^2}{2\pi}\right]^{d/2-1}
		\bigg(\check{H}^{(1)}_{d/2-1}\left(\tilde{r}_{\imath\jmath}\right)\left[\mathds{1}_{3}-\hat{\bm r}_{\imath\jmath}\hat{\bm r}_{\imath\jmath}\right]
		\nonumber\\&
		-\check{H}^{(1)}_{d/2}\left(\tilde{r}_{\imath\jmath}\right)\left[\mathds{1}_d-d\hat{\bm r}_{\imath\jmath}\hat{\bm r}_{\imath\jmath}\right]\bigg).
		\end{align}
		
		We now obtain~$\Gamma_{\imath\jmath}$ directly from Eqs.~(\ref{eq:omegaA}) and~(\ref{eq:gammaij}).
		Substituting
		\begin{equation}
		-k_0^2\hat{\bm k}\hat{\bm k}\text{e}^{\text{i}\omega\hat{\bm k}\cdot\bm{r}_{\imath\jmath}/c}
		=\nabla_d\nabla_d \text{e}^{\text{i}\hat{\bm k}\cdot\bm{r}_{\imath\jmath}\omega_0/c},\;
		\bm{r}_{\imath\jmath}\neq\bm0,
		\end{equation}
		into Eq.~(\ref{eq:gammaij}),
		and using
		\begin{equation}
		\int\text{d}^{d-1}\Omega_{\hat{\bm{k}}} \text{e}^{\text{i}\hat{\bm k}\cdot\bm{r}_{\imath\jmath}\omega_0/c}
		=\left(2\pi\right)^{d/2}\check{J}_{d/2-1}\left(\tilde{r}_{\imath\jmath}\right),
		\end{equation}
		yields
		\begin{equation}
		\gamma_{\imath\jmath}
		=\frac{k_0^d}{(2\pi)^{d-2}}
		\bm{\mu}_\imath\cdot\mathcal{D}\left[(2\pi)^{d/2}\check{J}_{d/2-1}\left(\tilde{r}_{\imath\jmath}\right)\right]\cdot\bm{\mu}_{\jmath}.
		\label{eq:gammaij2}
		\end{equation}
		Similarly,
		\begin{align}
		\omega_{\imath\jmath}
		=\frac{1}{2}\frac{k_0^d}{(2\pi)^{d-2}}
		\bm{\mu}_\imath\cdot\mathcal{D}\left[(2\pi)^{d/2}\check{Y}_{d/2-1}\left(\tilde{r}_{\imath\jmath}\right)\right]
		\cdot\bm{\mu}_{\jmath}.
		\label{eq:omegaij2}
		\end{align}
		Comparing Eqs.~(\ref{eq:gammaij2}) and~(\ref{eq:omegaij2})
		with~(\ref{eq:GdH}) proves the result.
	\end{proof}
	Equation~(\ref{eq:Gamma}) is a unified solution of collective atom-atom couplings for $d$D, and includes the previous results for $1$D~\cite{Lalumiere2013}, $2$D~\cite{Maximo2015}, and $3$D~\cite{Lehmberg1970}.
	Now we separate the terms governing the separation and orientation dependence of the collective atom-atom coupling by rewriting Eq.~(\ref{eq:Gamma}) as
	\begin{equation}
	\label{eq:GammaHTheta}
	\Gamma_{\imath\jmath}
	=\frac{\rm i}{2}\frac{\mu_\imath\mu_{\jmath} k_0^d}{(2\pi)^{d/2-2}}
	\left(\check{H}_{d/2-1}^{(1)}\left(\tilde{r}_{\imath\jmath}\right)\Theta_{\imath\jmath}
	-\check{H}_{d/2}^{(1)}\left(\tilde{r}_{\imath\jmath}\right)\Theta'_{\imath\jmath}\right)
	\end{equation}
	for
	\begin{align}
	\Theta_{\imath\jmath}
	=&\hat{\bm{\mu}}_\imath\cdot\hat{\bm{\mu}}_{\jmath}
	-(\hat{\bm{\mu}}_\imath\cdot\hat{\bm{r}}_{\imath\jmath})(\hat{\bm{\mu}}_{\jmath}\cdot\hat{\bm{r}}_{\imath\jmath}),\label{eq:thetaff}
	\\
	\Theta'_{\imath\jmath}
	=&\hat{\bm{\mu}}_\imath\cdot\mathds1_d\cdot\hat{\bm{\mu}}_{\jmath}
	-d(\hat{\bm{\mu}}_\imath\cdot\hat{\bm{r}}_{\imath\jmath})
	(\hat{\bm{\mu}}_{\jmath}\cdot\hat{\bm{r}}_{\imath\jmath}).
	\label{eq:thetanf}
	\end{align}
	Here the cardinal Hankel functions express the separation dependence of the collective effects, whereas~(\ref{eq:thetaff}) and~(\ref{eq:thetanf})
	summarize the orientation dependence of these effects.
	Asymptotically $\tilde{r}_{\imath\jmath}\gg 1$,
	\begin{equation}
	\label{eq:asymptote_far}
	\check{H}_{d/2-1}^{(1)}\left(\tilde{r}_{\imath\jmath}\right)
	\to\frac{\exp\left\{\text{i}
		\left[\tilde{r}_{\imath\jmath}-\frac{\pi}{4}(d-1)\right]\right\}}
	{\tilde{r}_{\imath\jmath}^{\frac{d-1}{2}}},
	\end{equation}
	leading to
	$\check{H}_{d/2-1}^{(1)}\left(\tilde{r}_{\imath\jmath}\right)/\check{H}_{d/2}^{(1)}\left(\tilde{r}_{\imath\jmath}\right)
	\to\text{i}\tilde{r}_{\imath\jmath}$,
	which shows that the first term in Eq.~(\ref{eq:GammaHTheta}) dominates for $\tilde{r}_{\imath\jmath}\gg 1$
	(defined here as far field)
	and the second term in Eq.~(\ref{eq:GammaHTheta}) which typically dominates for near field,
	defined as $\tilde{r}_{\imath\jmath}\ll 1$. We see that the near- and far-field terms are $\pi/2$ out of phase, so it is possible to use orientation control to suppress either $\gamma_{\imath\jmath}$ or $\omega_{\imath\jmath}$ by a factor of $\tilde{r}_{\imath\jmath}$ for distant atoms.
	
	Now we examine angular dependence of $\Gamma_{\imath\jmath}$~(\ref{eq:GammaHTheta})
	by studying the properties of
	$\Theta_{\imath\jmath}$~(\ref{eq:thetaff})
	and $\Theta'_{\imath\jmath}$~(\ref{eq:thetanf}).
	We restrict to parallel dipoles
	($\hat{\bm\mu}_\imath=\hat{\bm\mu}_\jmath$)
	separated along the~$x_1$ axis
	($\hat{\bm r}_{\imath\jmath}=\hat{\bm x}_1$)
	to visualize the angular dependence. In the far-field, the angular dependence is governed by the $d$-independent term
	$\Theta_{\imath\jmath}=1-\left(\hat{\bm\mu}_\imath\cdot\hat{\bm x}_1\right)^2$.
	Setting $\hat{\bm\mu}_\imath=\sin\theta_1\cos\theta_2\hat{\bm x}_1+\sin\theta_1\sin\theta_2\hat{\bm x}_2+\cos\theta_1\hat{\bm x}_3$
	yields $\Theta_{\imath\jmath}=1-\sin^2\theta_1\cos^2\theta_2$,
	which is a torus.
	In the near field, $\Gamma$ becomes $d$-dependent with   
	\begin{equation}
	\Theta'_{\imath\jmath}
	=\begin{cases}
	0,&d=1,\\
	-\sin^2\theta_1\cos 2\theta_2,&d=2,\\
	1-3\sin^2\theta_1\cos^2\theta_2,&d=3.
	\end{cases}
	\end{equation}
		
	We plot real and imaginary parts of~$\Gamma_{\imath\jmath}$
	in Fig.~\ref{fig:orientation}
	\begin{figure}	
		\includegraphics[width=0.94\columnwidth]{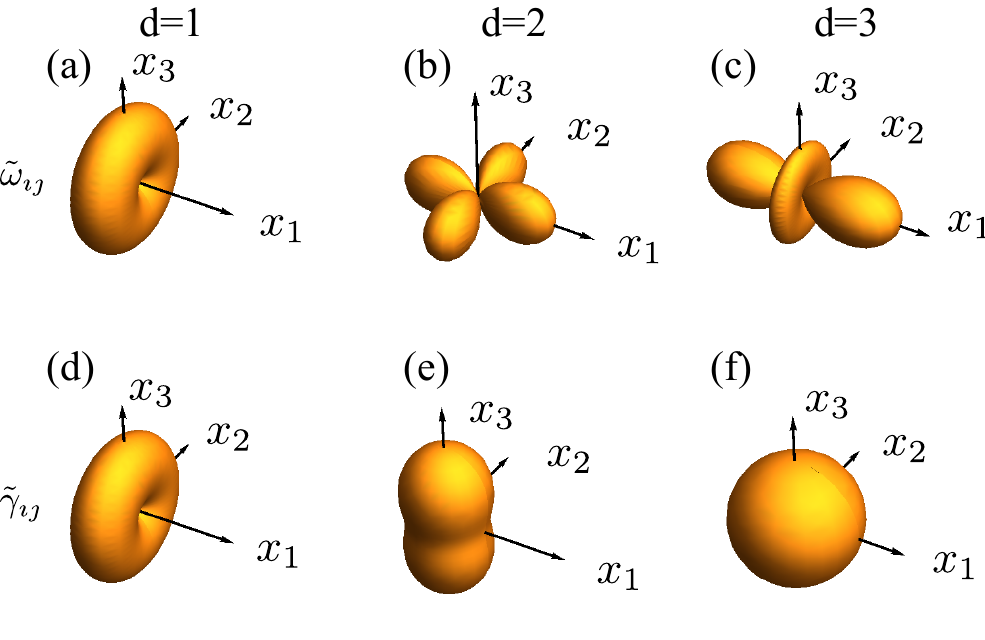}
		\caption{%
			Spherical polar plots of dimensionless ~$\tilde{\omega}_{12}=\omega_{12}/\gamma_{11}$ (a)-(c) and~$\tilde{\gamma}_{12}=\gamma_{12}/\gamma_{11}$ (d)-(f) up to a multiplicative constant for parallel dipoles $\bm{\mu}_1=\bm{\mu}_2=\sum_{l=1}^{3} x_l\hat{\bm{x}}_l$, $r_{12}\ll\lambda$, and $\hat{\bm{r}}_{12}=\hat{\bm{x}}_{1}$.}
		\label{fig:orientation}
	\end{figure}
	for parallel dipoles as functions of dipole orientation $\hat{\bm{\mu}}_{\imath}$ given by $x_1$,$x_2$,$x_3$. %The real and imaginary parts of~$\Gamma_{\imath\jmath}$ depend on both~$\Theta_{\imath\jmath}$ and~$\Theta'_{\imath\jmath}$ as well as 
    The interatomic separation is fixed to be very small
	($\tilde{r}_{\imath\jmath}\ll 1$)
	in order to correspond to the Dicke limit.
	The cylindrical symmetry of~$\omega_{\imath\jmath}$ for the 1D and 3D cases,
	as seen in Fig.~\ref{fig:orientation}(a,c),
	is replaced the four-leaf structure in 2D shown in Fig.~\ref{fig:orientation}(b),
	and the simple plot of~$\gamma_{\imath\jmath}$ in Fig.~\ref{fig:orientation}(f)
	transforms to more complicated surfaces in Fig.~\ref{fig:orientation}(d,e)
	due to enhanced emission for atoms oriented perpendicular to its confinement.
	
	Collective effects~(\ref{eq:GammaHTheta})
	are strongly dependent on dimensional confinement,
	as evidenced by the contrast between inverse-distance dependence in 3D vs constant in 1D
	for large separation $\tilde{r}_{\imath\jmath}\gg1$~\cite{Lehmberg1970,Lalumiere2013}.
	The $d$-dependence of~$\Gamma_{\imath\jmath}$
	is captured by the asymptotic expression for the cardinal Hankel function~(\ref{eq:asymptote_far})
	whose denominator shows $d$-dependent fall-off
	and whose oscillatory exponential numerator
	shows that~$\gamma_{\imath\jmath}$ and~$\omega_{\imath\jmath}$
	are $\pi/2$ out of phase.
	Furthermore $\check{H}_{d/2-1}^{(1)}$ experiences a~$\pi/4$ phase shift for each integer leap in dimension~$d$, corresponding to a $\lambda_0/8$ shift in relative positions of the atoms in different dimensions for maximizing atom-field coupling.
	
	Whereas~$\omega_{\imath\jmath}$
	and~$\gamma_{\imath\jmath}$
	display similar features for well separated parallel dipoles,
	the closely spaced parallel-dipole case
	is quite different
	due to $\gamma_{\imath\jmath}$ being sensitive to both near- and far-field terms in (\ref{eq:GammaHTheta}) while $\omega_{\imath\jmath}$ is only sensitive to near field terms.	Specifically,
	the asymptotic expressions for the cardinal Bessel functions yield
	$\gamma_{\imath\jmath}\mapsto\gamma_{\imath\imath}
	\left[1-\mathcal{O}(\tilde{r}_{\imath\jmath}^{2})\right]$,
	which is independent of~$d$,
	whereas
	\begin{equation}
	\omega_{\imath\jmath}
	\sim\begin{cases}
	\tilde{r}_{\imath\jmath}^{-d},&\Theta'_{\imath\jmath}\neq0,\\
	\tilde{r}_{\imath\jmath}^{-d+2},&\Theta'_{\imath\jmath}=0,d\neq2,\\
	\log\tilde{r}_{\imath\jmath},&\Theta'_{\imath\jmath}=0,d=2.
	\end{cases}
	\end{equation}
	We now have asymptotic expressions of~$\gamma_{\imath\jmath}$
	and~$\omega_{\imath\jmath}$ in the asymptotic small and large~$\tilde{r}_{\imath\jmath}$
	regimes and now explore the dependence on the full range of~$\tilde{r}_{\imath\jmath}$.
	
	We plot each of~$\omega_{\imath\jmath}$ and~$\gamma_{\imath\jmath}$
	as a function of both~$\tilde{r}_{\imath\jmath}$ and~$d$
	as surface plots in Fig.~\ref{fig:normalizedGamma}(a,c)
	and present slices of those plots in Fig.~\ref{fig:normalizedGamma}(b,d).
	\begin{figure}	
		\includegraphics[width=0.94\columnwidth]{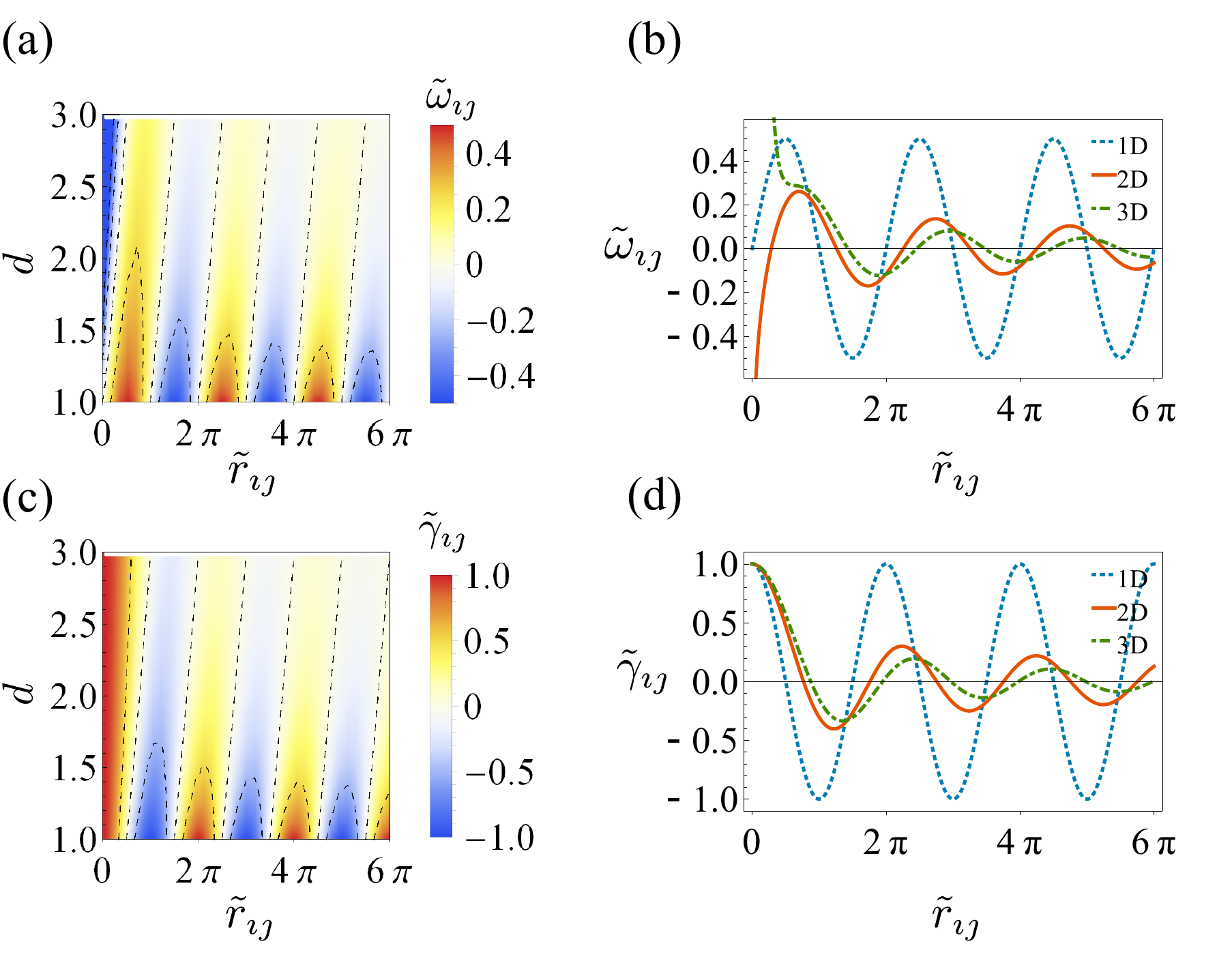}
		\caption{%
			(Color online) Dimensional and separation dependence of dimensionless $\tilde{\omega}_{\imath\jmath}:=\omega_{\imath\jmath}/\gamma_{\imath\imath}$ ((a)-(b)) and $\tilde{\gamma}_{\imath\jmath}:=\gamma_{\imath\jmath}/\gamma_{\imath\imath}$ ((c)-(d))
			vs dimensionless separation~$\tilde{r}_{\imath\jmath}=2\pi\frac{r_{\imath\jmath}}{\lambda}$ 
for identical parallel dipoles $\hat{\bm{\mu}}_\imath=\hat{\bm{\mu}}_{\jmath}=\hat{\bm{x}}_3$.
(a) and (c) show results interpolated for real valued dimensions $1\leq d\leq3$. (b) and (d) compare $d=1$ (dotted blue line), $d=2$ (solid red line), $d=3$ (dot-dashed green line).%
		}
		\label{fig:normalizedGamma}
	\end{figure}
	We have interpolated between integer dimensions by inserting the
	modified identity
	\begin{equation}
	\mathds{1}_d
	=\sum_{l=1}^{\lceil d \rceil }\hat{\bm{x}}_l\hat{\bm{x}}_l
	+\left(d-\lceil d \rceil\right)\hat{\bm{x}}_{\lceil d \rceil}
	\hat{\bm{x}}_{\lceil d \rceil}
	\end{equation}
	into Eq.~(\ref{eq:thetanf}),
	where $\lceil~\rceil$ is the ceiling function.
	The small and large~$\tilde{r}_{\imath\jmath}$ features have been explained already,
	and the plot shows that these small and large limits apply everywhere except a small region near $\tilde{r}_{\imath\jmath}\sim1$.
	Interestingly our $d$-dependent functions are smooth for real-valued~$d$,
	thus giving us clear predictions of collective behavior for non-integer dimension.
	Exploration of non-integer~$d$ collective effects would be quite interesting
	and could relate to electromagnetic field anisotropy~\cite{He1991}.

	As $1$D and $3$D collective effects have been explored experimentally,
	we propose a $2$D experiment with vacancy centers in diamond as our ``atoms''. In addition to requiring a structure that confines the electromagnetic field to $2$D, we have three requirements for the emitters for realizing $2$D superradiance: sub-wavelength relative position control, lifetime-limited linewidths, and spectrally overlapping energies. The $2$D structure and emitter-position control ensure the ability to control superradiance phenomena, while the spectral requirements are necessary for their observation.
	
	 There are two promising approaches towards a $2$D diamond structure: ultra-high aspect ratio diamond thinned via plasma etching~\cite{Tao2013} and membrane structures of sub-wavelength thicknesses~\cite{Piracha2016}. As the diamond medium is not the vacuum described thus far, we extend our result to dielectric media using~\cite{Knoester1989,Barnett1992}
	\begin{align}	\Gamma_{\imath\jmath,\epsilon(\omega)}(r_{\imath\jmath})=\text{Re}\left[\epsilon(\omega_0)^{1/2}\right]|l|^{2}\Gamma_{\imath\jmath}\left(\text{Re}\left[\epsilon(\omega_0)^{1/2}\right]r_{\imath\jmath}\right),
	\end{align}
	where $\epsilon(\omega)$ is the dielectric coefficient, and $l$ is a local electric field factor.
	
	To satisfy the requirements on the emitters, ion implantation techniques allow either nitrogen or silicon vacancies to be positioned with impressive $r_{\imath\jmath}\sim\lambda_0/20$ accuracy~\cite{Toyli2010,Schukraft2016,Sipahigil2016,Rogers2014}. We propose working with a single pair of vacancies as shown in Fig.~\ref{fig:schematic}(a) to minimize inhomogeneity inherent in an ensemble. Nitrogen vacancy centers are appealing due to their narrow homogeneous linewidths~\cite{Santori2010} but suffer from strain-induced inhomogeneous broadening that can be ameliorated by Stark shifting from an external field \cite{Tamarat2006}.% displayed by electrodes in Fig.~\ref{fig:schematic}(a). 
In contrast, silicon vacancies have inversion symmetry that protects them from external fields, thereby reducing inhomogeneity but makes spectral control via Stark shifts challenging \cite{Rogers2014}. However each silicon vacancy can be addressed with a tunable off-resonant laser to obtain spectrally overlapping Raman transitions, as has been used to demonstrate $1$D superradiance \cite{Sipahigil2016}.

	For either nitrogen- or silicon- vacancy centers, the pair can be excited symmetrically by a resonant pulse with bandwidth much less than $\gamma_{\imath\imath}$ and propagating perpendicular to $\bm{r}_{\imath\jmath}$. Superradiant effects can be quantified by $\omega_{\imath\jmath}(\tilde{r}_{\imath\jmath}, \bm{\mu}_{\imath},\bm{\mu}_{\jmath})$ and $\gamma_{\imath\jmath}(\tilde{r}_{\imath\jmath}, \bm{\mu}_{\imath},\bm{\mu}_{\jmath})$ through time-resolved photoluminescence measurements as outlined in Fig.~\ref{fig:schematic}(b).
	
	In conclusion, we present a unified solution for collective spontaneous emission, for electromagnetic field confined to dimension $d\in[1,2,3]$, with arbitrary dipole orientation and separation. We explain the scaling behavior of cooperative effects for systems much larger or smaller than the resonance wavelength. Furthermore we suggest a potential implementation scheme using vacancy centers in diamond to explore the effects in $2$D.
	
	We thank Paul Barclay for valuable discussions. TAH and HD acknowledge support from NSF Grant DMR-1120923 and DMR-1150593, and from AFOSR Grant FA9550-15-1-0240.
	BCS acknowledges support from NSERC, Alberta Innovates, China's 1000 Talent Plan, the Institute for Quantum Information and Matter, an NSF Physics Frontiers Center (NSF Grant PHY-1125565), and the support of the Gordon and Betty Moore Foundation (GBMF-2644).

	\bibliography{superrad}
\end{document}